\documentclass[reprint, pra, superscriptaddress]{revtex4-2}

\usepackage[utf8]{inputenc}
\usepackage[english]{babel}
\usepackage{fullpage}
\usepackage{parskip}
\usepackage{tikz}
\usepackage{amsmath,amsthm,amsfonts}
\usepackage{mathtools}
\usepackage{verbatim}
\usepackage{hyperref}
\usepackage{amssymb}
\usepackage{graphicx}

\newtheorem{theorem}{Theorem}[section]

\newtheorem{proposition}[theorem]{Proposition}
\newtheorem{corollary}[theorem]{Corollary}
\theoremstyle{definition}

\newcommand{\mbf}{\mathbf}
\newcommand{\mbb}{\mathbb}
\newcommand{\mc}{\mathcal}

\newcommand{\tr}{\textrm{Tr}}
\newcommand{\wt}{\widetilde}

\newcommand{\ket}[1]{|#1\rangle}

\newcommand{\op}[2]{|#1\rangle\langle #2|}
\newcommand{\ip}[2]{\langle #1| #2\rangle}

\definecolor{cool_green}{rgb}{0.0, 0.5, 0.0}

\begin{document}

\title{Entanglement of assistance in three-qubit systems}

\author{Kl\'{e}e Pollock}
\affiliation{Department of Electrical and Computer Engineering, Coordinated Science Laboratory, University of Illinois at Urbana-Champaign, Urbana, Illinois 61801, USA}
\affiliation{Department of Physics and Astronomy, Iowa State University of Science and Technology, Ames, Iowa 50014-3160, USA}
\author{Ge Wang}
\affiliation{Department of Physics, University of Illinois at Urbana-Champaign, Urbana, IL 61801, USA}
\author{Eric Chitambar}
\affiliation{Department of Electrical and Computer Engineering, Coordinated Science Laboratory, University of Illinois at Urbana-Champaign, Urbana, Illinois 61801, USA}

\date{\today}

\begin{abstract}
The entanglement of assistance quantifies the amount of entanglement that can be concentrated among a group of spatially-separated parties using the assistance of some auxiliary system.  In this paper we study the entanglement of assistance in the simplest scenario of three qubits.  We consider how much entanglement is lost when the helper party becomes uncoupled by local measurement and classical communication.  For three-qubit pure states, it is found that lossless decoupling is always possible when entanglement is quantified by the maximum probability of locally generating a Bell state.  However, with respect to most other entanglement measures, such as concurrence and entanglement entropy, a lossless decoupling of the helper is possible only for a special family of states that we fully characterize.  
\end{abstract}

\maketitle

\section{Introduction}

A quantum system prepared in a pure state is uncoupled from any other system in the universe.  Operationally, this means that if $p(a)$ is the outcome probability distribution when measuring a pure state of system $A$ and $p(c)$ is the outcome distribution when measuring any other system $C$, then their joint distribution will necessarily factorize: $p(a,c)=p(a)p(c)$.  When this fact is extended to pure states of bipartite systems, a monogamy relationship follows \cite{Werner-1989a, Coffman-2000a, Terhal-2004a}.  That is, if systems $A$ and $B$ are in some pure entangled state, then $C$ must be uncoupled from them both.  In particular, $A$ and $C$ cannot also be entangled.

\begin{figure}[h]
\includegraphics{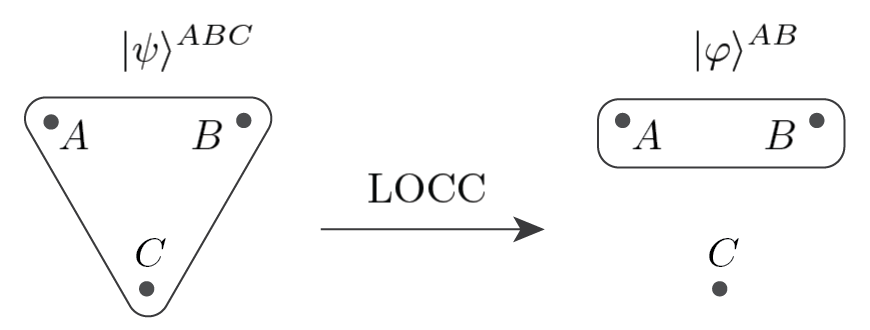}
\caption{A pure-state decoupling transformation.}
\label{fig:decoupling}
\end{figure}

While the ``correlation shielding'' provided by the monogamy of entanglement is essential for performing cryptographic tasks like quantum key distribution \cite{Lo-1999a, Vazirani-2014a}, it creates fundamental challenges to quantum information processing using multipartite entangled states.  As the simplest case, consider the three-node quantum network depicted in Fig. \ref{fig:decoupling}.  Suppose the parties stationed at each node share some tripartite entangled pure state $\ket{\psi}^{ABC}$, and they wish to transform it into a bipartite entangled state $\ket{\varphi}^{AB}$ held by parties Alice ($A$) and Bob ($B$).  Due to the monogamy of entanglement, Charlie ($C$) must become completely uncoupled from both Alice and Bob in the transformation  $\ket{\psi}^{ABC}\to \ket{\varphi}^{AB}$.  In other words, the final state is implicitly understood to have the tripartite form $\ket{\varphi}^{AB}\ket{0}^C$ (or more generally $\op{\varphi}{\varphi}^{AB}\otimes\rho^C$).  The uncoupling of Charlie is not modeled simply by tracing out his system, i.e. $\op{\varphi}{\varphi}^{AB}\not=\tr_C\op{\psi}{\psi}^{ABC}$, and more elaborate protocols are needed to achieve transformations of the form $\ket{\psi}^{ABC}\to \ket{\varphi}^{AB}$.  

In one scenario, Charlie could be seen as an unwanted party, either an eavesdropper or a noisy environment.  With Charlie being adversarial, the decoupling protocol would be driven exclusively by Alice and Bob, and the problem then reduces to the task of entanglement purification \cite{Deutsch-1996a, Bennett-1996a}.  In another scenario, Charlie is seen as a helper, and he provides assistance to Alice and Bob with the state transformation.  This assistance is provided by Charlie making local measurements on his subsystem and communicating his outcome to Alice and Bob.  Due to the stochastic nature of quantum measurements, it becomes more appropriate to consider multi-outcome transformations, in which Charlie is decoupled for each possible outcome.  We can denote such a process as $\ket{\psi}^{ABC}\to\ket{\varphi_i}^{AB}$, where $\ket{\varphi_i}^{AB}$ is obtained with probability $p_i$. 

The amount of help that Charlie can provide in generating bipartite pure-state entanglement for Alice and Bob is known as the Entanglement of Assistance (EoA) \cite{DiVincenzo-1999a}.  It can be defined as the maximum average entanglement, $\langle E_{A|B}\rangle:=\sum_i p_i E(\varphi_i)$, over all transformations $\ket{\psi}^{ABC}\to\ket{\varphi_i}^{AB}$ in which Charlie simply decouples himself by performing a single measurement.  Here, $E(\varphi_i)$ is the entanglement of state $\ket{\varphi_i}^{AB}$ for some chosen entanglement monotone $E$.  Note, the measurement of Charlie's which maximizes EoA for one monotone may not be the same as that which maximizes EoA for another, and we will see an example of this below.  Nevertheless, all entanglement monotones satisfy the monotonicity constraint
\begin{equation}
\label{Eq:EoA-bipartite-bound}
\langle E_{A|B}\rangle \leq \min\{E_{A|BC}(\psi),E_{B|AC}(\psi)\},
\end{equation}
where $E_{A|BC}(\psi)$ is the bipartite entanglement in the initial state $\ket{\psi}^{ABC}$ when parties $B$ and $C$ are grouped together, and similarly for $E_{B|AC}(\psi)$.  Ineq. \eqref{Eq:EoA-bipartite-bound} holds since monotonicity prevents both the $A|BC$ entanglement and the $B|AC$ entanglement from increasing on average.

For a given tripartite entangled state $\ket{\psi}^{ABC}$, a natural question is whether Charlie can perform a measurement that decouples him in a ``lossless'' fashion.  This means that there exists a multi-outcome transformation $\ket{\psi}^{ABC}\to\ket{\varphi_i}^{AB}$ such that Ineq.  \eqref{Eq:EoA-bipartite-bound} is tight, i.e. $\langle E_{A|B}\rangle = \min\{E_{A|BC}(\psi),E_{B|AC}(\psi)\}$.  Remarkably, it has been shown that for asymptotic pure state distillation, an equality in Ineq. \eqref{Eq:EoA-bipartite-bound} indeed holds after the LHS is replaced by the optimal rate of bipartite entanglement yield \cite{Smolin-2005a}.  The process of lossless asymptotic decoupling has been generalized to the multi-party setting in the form of entanglement combing \cite{Yang-2009a}.  However, in the single-copy setting, Ineq. \eqref{Eq:EoA-bipartite-bound} is generally strict.

The goal of this paper is to investigate which three-qubit states $\ket{\psi}^{ABC}$ can achieve an equality in Ineq. \eqref{Eq:EoA-bipartite-bound}.  The answer to this question depends on the choice of monotone $E$.  Our main result is showing that every three-qubit pure state can achieve $\langle E_2^{A|B}\rangle = \min\{E_2^{A|BC}(\psi),E_2^{B|AC}(\psi)\}$ for some measurement of Charlie, where $E_2(\varphi^{AB})$ is twice the smallest eigenvalue of $\varphi^A=\tr_B\op{\varphi}{\varphi}^{AB}$ \cite{Vidal-1999b}.  Operationally this quantity is important because it corresponds to the maximum probability of converting $\ket{\varphi}^{AB}$ into the maximally entangled state $\ket{\Phi^+}=\sqrt{1/2}(\ket{00}+\ket{11})$ using LOCC \cite{Lo-2001a}.  Hence our result implies that the optimal probability of converting any tripartite pure state $\ket{\psi}^{ABC}$ into a bipartite maximally-entangled state between $AB$ is the same whether $C$ acts either locally or jointly with one of the other parties.  It also implies that the EoA for $E_2$ is a genuine entanglement monotone in three qubits, a fact which does not hold for entanglement monotones in general \cite{Gour-2006a}.  We prove this result by showing that a sufficient condition for lossless $E_2$ decoupling presented in Ref. \cite{Gour-2005a} actually applies to all three-qubit states.  We then consider other entanglement monotones and provide the complete family of states for which $\langle E^{A|B}\rangle = \min\{E^{A|BC}(\psi),E^{B|AC}(\psi)\}$ with respect to any entanglement monotone that is a strictly concave function of the square-Schmidt coefficients \cite{Vidal-2000a}.  The latter includes the entanglement entropy and the concurrence \cite{Wootters-1998a}, to name a few.

\section{Definitions and Preliminaries}

\subsection{Bipartite Entanglement Measures}

The discovery that the nonlocal correlations generated by quantum entanglement are useful for communication and information processing has motivated an effort to precisely quantify entanglement.  Various entanglement measures have been proposed in the literature, and all of them satisfy the necessary property that their value cannot be increased under LOCC \cite{Horodecki-2009a}.  That is for any entanglement measure $E$, if $\mc{L}$ is a completely-positive trace-preserving (CPTP) map generated by a protocol of local operations and classical communication between $N$ parties, then $E\left(\rho\right)\geq E\left(\mc{L}(\rho)\right)$ for all $N$-partite states $\rho$.  It may be the case that $\mc{L}$ outputs a classical register (or ``flag'') storing various amounts of classical data or measurement outcomes $x$ obtained in the protocol.  This requires that
\begin{equation}
\label{Eq:entanglement-monotone1}
    E\left(\rho\right)\geq E\left(\sum_x p_x \rho_x\otimes\op{x}{x}\right),
\end{equation}
where $\rho_x$ is the quantum state conditioned on obtaining classical data $x$ and $p_x$ is the probability of this event.  States on the RHS of \eqref{Eq:entanglement-monotone1} are known as qc (quantum-classical) states.  Most entanglement measures of interest are convex-linear on qc states, meaning that $E\left(\sum_x p_x \rho_x\otimes\op{x}{x}\right)=\sum_x p_xE(\rho_x\otimes\op{x}{x})=\sum_x p_xE(\rho_x)$.  From this we see
\begin{equation}
\label{Eq:entanglement-monotone2}
    E\left(\rho\right)\geq \sum_x p_x E\left(\rho_x\right),
\end{equation}
and any function satisfying this relationship is known as an entanglement monotone.  Throughout this paper we will only consider measures that are convex-linear on qc states, and so Eq. \eqref{Eq:entanglement-monotone2} can make no distinction between the terms entanglement measure and monotone in terms of their behavior under LOCC.

The general structure of bipartite entanglement monotones has been analyzed by Vidal in Ref. \cite{Vidal-2000a}.  For LOCC entanglement transformations in which the initial and final states are pure, it suffices to consider monotones that are defined just on pure states.  Vidal shows that every pure state entanglement monotone is determined by a concave, symmetric function of the squared-Schmidt coefficients of the given state $\ket{\varphi}^{AB}$ (equivalently of the eigenvalues of the reduced density matrix $\varphi^A$).  For example, the Entanglement Entropy  $S_1(\varphi^{AB})=-\tr[\varphi^A\log\varphi^A]$, and more generally the $\alpha$-Entanglement Entropy $S_{\alpha}(\varphi^{AB})=\frac{1}{1-\alpha}\log{\tr{(\varphi^A)^{\alpha}}}$ for $\alpha\in[0,1]$, all arise from a concave, symmetric function $g : \Delta _n \rightarrow \mbb{R}^+$ on the $n$-dimensional unit simplex $\Delta_n$.  Another family of monotones are the so-called Ky-Fan functions $\{E_k\}_{k=1}^d$ \cite{Vidal-1999b}, which are given by $E_k(\varphi^{AB})=\sum_{n=k}^d\lambda^{\downarrow}_k$, where the $\lambda^{\downarrow}_k$ are the eigenvalues of $\varphi^A$ arranged in decreasing order and $d=\min\{\dim(\varphi^A),\dim(\varphi^B)\}$.  Note that $E_k$ generalizes the $E_2$ measure mentioned above.  A final example is the generalized concurrence (G-concurrence) $G(\varphi^{AB})=d\det[\varphi^A]^{\frac{1}{d}}$ \cite{PhysRevA.71.012318}, which for the two-qubit case ($d=2$) reduces to the original concurrence \cite{Wootters-1998a}, as well as the robustness of entanglement \cite{Vidal-1999a}.

\subsection{The Entanglement of Collaboration and Assistance}

As depicted in Fig. \ref{fig:decoupling}, the general problem we consider is the decoupling of some party $C$ from parties $AB$ when they all share an entangled state $\ket{\psi}^{ABC}$.  Under locality constraints, the most general procedure they can perform is three-way LOCC with multiple rounds of classical communication exchanged between all the parties.  The maximum average bipartite entanglement they generate under such a process is called the entanglement of collaboration (EoC) \cite{Gour-2006a, Gour-2006b}, and for a bipartite entanglement measure $E$ it is given by
\begin{align}
\label{Eq:EoC-definition}
    E^{(c)}(\psi^{ABC})=\sup_{\mc{L}}\sum_x p_xE(\varphi_x^{AB}),
\end{align}
with the supremum taken over all three-way LOCC maps $\mc{L}$ transforming $\ket{\psi}^{ABC}$ to $\ket{\varphi_x}^{AB}$ with probability $p_x$.  Note that EoC is an entanglement monotone for tripartite pure states since it is defined in terms of an optimization over LOCC protocols \cite{Vidal-2000a}.

In this paper we are interested in more restricted decoupling protocols in which only one-way communication from the helper is allowed. In particular we consider a protocol where Charlie (the helper) first makes a decoupling measurement on his part of $\ket{\psi}^{ABC}$, and then broadcasts his measurement outcome to both Alice and Bob. In this process, Charlie's measurement ``collapses'' the tripartite state into some bipartite entangled state shared by Alice and Bob, the form of which they are made aware when they receive Charlie's classical information.  Alice and Bob can then proceed with bipartite LOCC processing of their post-measurement state.

% A more restricted process of decoupling, and one of primary interest in this paper, involves just one-way LOCC protocols with communication being strictly from Charlie (the helper) to both Alice and Bob.  This can be understood as Charlie making a decoupling measurement on his part of $\ket{\psi}^{ABC}$ and thereby ``collapsing'' the state into some bipartite entangled state shared by Bob and Charlie, the form of which they know thanks to the classical communication.
The Entanglement of Assistance is defined with respect to these types of protocols.   If $\vec{\mc{L}}$ denotes the set of one-way LOCC maps from Charlie to Alice and Bob, the EoA is given by
\begin{align}
\label{Eq:EoA-definition}
    E^{(a)}(\psi^{ABC})=\max_{\vec{\mc{L}}}\sum_x p_xE(\varphi_x^{AB}).
\end{align}
From Ineq. \eqref{Eq:EoA-bipartite-bound}, one can immediately see that any bipartite entanglement monotone $E$ satisfies the hierarchy
\begin{equation}
\label{Eq:EoA-bipartite-bound2}
    E^{(a)}(\psi)\leq  E^{(c)}(\psi)\leq \min\{E^{A|BC}(\psi),E^{B|AC}(\psi)\}.
\end{equation}

In this paper we are viewing the entanglement of assistance as a function of tripartite pure states \cite{Gour-2006a}.  However, this is not the only perspective one could take, and originally \cite{DiVincenzo-1999a} the EoA was introduced as a function of bipartite mixed states given by
\begin{equation}
\label{Eq:EoA-definition2}
    \wt{E}^{(a)}(\rho^{AB})=\max_{\{p_k,\ket{\varphi_k}\}_k} \sum_k p_k E(\varphi_k),
\end{equation}
where the maximization is taken over all pure state ensembles $\{p_k,\ket{\varphi_k}\}_k$ such that $\rho^{AB}=\sum_k p_k\op{\varphi_k}{\varphi_k}$.  Since every possible pure state ensemble for $\rho^{AB}$ can be realized as the post-measurement ensemble generated by some measurement of Charlie on any purification of $\ket{\psi}^{ABC}$ or $\rho^{AB}$ (i.e. $\rho^{AB}=\tr_C\op{\psi}{\psi}^{ABC}$) \cite{Hughston-1993a}, we clearly have
\begin{equation}
\label{Eq:EOA-defns}
    E^{(a)}(\psi^{ABC})=\wt{E}^{(a)}(\rho^{AB}).
\end{equation}
The reason we prefer to interpret EoA as a tripartite pure-state measure is because its operational meaning refers to a process conducted on three systems.  Stated differently, if an experimenter only has access to the bipartite state $\rho^{AB}$ and not a purifying third system, then the physical meaning of $\wt{E}^{(a)}(\rho^{AB})$ is not clear since it no longer describes how much entanglement can be concentrated between systems $A$ and $B$.

%where $\mc{E}$ runs over all pure state ensembles $\mc{E}=\{(p_k,\ket{\varphi_k})\}_k$ such that $\rho^{AB}=\sum_k p_k\op{\varphi_k}{\varphi_k}$, and $\langle E \rangle_{\mc{E}}=\sum_k p_k E(\varphi_k)$ is the ensemble average entanglement.  While Eq. \eqref{Eq:EoA-definition} expresses $E^{(a)}$ as a function of a bipartite mixed state, it should really be interpreted as a function of its tripartite purification $\ket{\psi}^{ABC}$, i.e. $\rho^{AB}=\tr_C\op{\psi}{\psi}^{ABC}$ \cite{Gour-2006a}.  Since all purifications are related by an isometry applied to system $C$, the particular choice of purification considered is irrelevant.  By viewing $E^{(a)}$ as a function of tripartite pure states, the operational meaning of EoA given in the introduction is equivalent to the definition given in Eq. \eqref{Eq:EoA-definition}.  Namely, every possible pure state ensemble for $\rho^{AB}$ can be realized as the post-measurement ensemble generated by some measurement of Charlie on the purified state $\ket{\psi}^{ABC}$ \cite{Hughston-1993a}.

Unlike the EoC, the EoA is not a tripartite entanglement monotone in general \cite{Gour-2006a}.  This means that when starting with some state $\ket{\psi}^{ABC}$, there exists LOCC protocols in which communication from Alice and Bob to Charlie before he makes his decoupling measurement can increase their average post-measurement entanglement. A known exception to this is for $2\otimes 2\otimes n$ systems when concurrence is used as the two-qubit entanglement measure \cite{Gour-2005a}, as well as for higher-dimensional tripartite systems when the $G$-concurrence is the bipartite entanglement quantifier \cite{Gour-2006b}.  A main result of this paper establishes that $E_2^{(a)}$ is also an entanglement monotone for three-qubit pure states, which is the EoA for $E_2$ entanglement.

%In general this quantity is not an entanglement measure for bipartite mixed states, but it is operationally significant in the following scenario. If Alice, Bob, and Charlie begin with a pure tripartite state $\ket{\psi}^{ABC}$, then $E^{(a)}$ is the average optimal pure state entanglement that Charlie can leave Alice and Bob with after decoupling himself by performing a measurement. This follows because 

\section{Main Results}

The central question we ask is when $E^{(a)}(\psi)=\min\{E^{A|BC}(\psi),E^{B|AC}(\psi)\}$ in Ineq. \eqref{Eq:EoA-bipartite-bound2} for a three-qubit state $\ket{\psi}$.  %As noted above, any two-qubit pure state entanglement monotone $E$ is given by $E(\varphi^{AB})=f(\lambda_{\min}(\varphi^{A}))$ with  $f:[0,\frac{1}{2}]\rightarrow \mbb{R}^+$ being concave and $\lambda_{\min}(\varphi^A)$ the smallest eigenvalue of $\varphi^A$.  Via Eq. \eqref{Eq:EOA-defns}, there exists an optimal ensemble $\{p_x,\ket{\varphi_x}\}_x$ of $\rho^{AB}$ so that
%\begin{align}
 %   E^{(a)}(\psi^{ABC})&=\sum_x p_xE(\varphi_x^{AB})\notag\\
  %  &=\sum_x p_xf(\lambda_{\min}(\varphi_x^{A}))\notag\\
   % &\leq f\left(\sum_x % p_x\lambda_{\min}(\varphi_x^A)\right)
%\end{align}
We begin our analysis by stating a structural result for three qubit pure states, which was proven in Ref. \cite{Gour-2005a} as a method for computing $E^{(a)}_2$.  In the following, we let $\ket{\Tilde{\Phi}^+}^{AB}$ denote the unnormalized maximally entangled state $\ket{00}+\ket{11}$ for some \textit{a priori} fixed bases of systems $A$ and $B$, and complex conjugation is understood to be taken with respect to these bases.
\begin{proposition}[\cite{Gour-2005a}]
\label{Prop:Gour}
Let $\ket{\psi}^{ABC}$ be an arbitrary $2\otimes 2\otimes n$ state.  Then there exists two orthonormal bases $\{\ket{k}^C\}_{k=1}^n$ and $\{\ket{k'}^C\}_{k=1}^n$ for system $C$ such that
\begin{align}
\ket{\psi}&=\sum_k A_k\otimes\mbb{I}^{BC}\ket{\Tilde{\Phi}^+}^{AB} \ket{k}^C
\label{Eq:basis-A}\\
\ket{\psi}&=\sum_k \mbb{I}^{AC}\otimes B_k\ket{\Tilde{\Phi}^+}^{AB} \ket{k'}^C\label{Eq:basis-B}
\end{align}
with both $\{A_kA_k^\dagger\}_k$ and $\{B_kB_k^\dagger\}_k$ being sets of pairwise commuting operators.  %Furthermore, if $E^{A|BC}(\psi)=E^{B|AC}(\psi)$ for some entanglement measure $E$, then $\{A^T_kA_k^*\}_k$ and $\{B^T_kB_k^*\}_k$ are also sets of of pairwise commuting operators.
\end{proposition}
%\begin{proof}
%The existence of bases $\{\ket{k}^C\}_{k=1}^n$ and $\{\ket{k'}^C\}_{k=1}^n$ having the stated properties is given in Ref. \cite{Gour-2005a}.  
%\end{proof}

Using this proposition, we prove our first result.
\begin{theorem}
\label{Thm:1}
For any three-qubit state $\ket{\psi}^{ABC}$
\begin{equation}
E_2^{(a)}(\psi) = \min \{  E_2^{A|BC}(\psi), E_2^{B|AC}(\psi)\}.
\end{equation}
\end{theorem}

\begin{proof}
%The authors in Ref. \cite{Gour-2005a} show that for any three-qubit state there always exists a basis $\ket{k}^C$ for Charlie's system such that we can write $\ket{\psi}^{ABC}=\sum_k A^k \otimes \mbb{I}^{BC}\ket{\Tilde{\Phi}^+}^{AB} \ket{k}^C$ with the operators $A^k A^{k \dagger}$ simultaneously diagonalizable, and $\ket{\Tilde{\Phi}^+} = \ket{00}+\ket{11}$. There also exists another basis $\ket{k'}^C$ for which we may write $\ket{\psi}^{ABC}=\sum_k A'^k \otimes \mbb{I}^{BC}\ket{\Phi^+}^{AB}\ket{k'}^C$ with the $A'^{k \dagger} A'^k$ now simultaneously diagonalizable. Our convention is to take $k\in\{1,2\}$. Let $\{p_k, \ket{\psi_k}^{AB}\}_k$ be the ensemble resulting from Charlie's projective measurement in the basis $\{\ket{k}^C\}_k$. If $p_k=0$ for some $k$, then $\rho^{AB}$ is pure, the state is a product state across the  partition AB$\vert$C, and $\mc{P}_{\text{max}}=E_{A|B}=E_{A|BC}=E_{B|AC}$. Assuming $p_k\neq0$, define
%\begin{equation}
%    \rho_k^A = \tr_B \op{\psi_k}{\psi_k} =\frac{1}{p_k}A^k A^{k \dagger}= \frac{1}{2}(\mbb{I}+\vec{r_k}\cdot\vec{\sigma})
%\end{equation}
%\begin{equation}
%    \sigma_k^B = \tr_A \op{\psi_k}{\psi_k} = \frac{1}{p_k}(A^{k \dagger} A^k)^* = \frac{1}{2}(\mbb{I}+\vec{s_k}\cdot\vec{\sigma}),
%\end{equation}
Assume that $\tr_C{\op{\psi}{\psi}}^{ABC}$ is not rank-one, or else Charlie is already decoupled and the statement trivially holds. Consider the operators $\{A_k\}_{k=1}^2$ and $\{B_k\}_{k=1}^2$ given by Eqns. \eqref{Eq:basis-A} and \eqref{Eq:basis-B} in Proposition \ref{Prop:Gour}, where we have taken $n=2$.  %Assume that both the $A_k$ are nonzero or else the theorem trivially holds.  Consider first the case that one of the $\{A_0,A_1,B_0,B_1\}$ is a unitary operator.  Then $\ket{\psi}$ will have the form
%\begin{equation}
%    \ket{\psi}=\cos\theta(U\otimes\mbb{I})\ket{\Phi+}\ket{0}+\sin\theta\ket{\varphi}\ket{1}.
%\end{equation}
For $\ket{\psi_k}:=(A_k\otimes \mbb{I})\ket{\tilde{\Phi}^+}$ and $p_k=\tr[A_kA_k^\dagger]$, we can write
\begin{align}
    \rho_k^A=\tr_B \op{\psi_k}{\psi_k}=\frac{1}{p_k}A_k A_k^{\dagger}= \frac{1}{2}(\mbb{I}+\mbf{r}_k\cdot\vec{\sigma}),
\end{align}
where $\mbf{r}_k$ are the Bloch vectors of the $\rho_k^A$.  Since $\rho_k^A$ and $\rho_k^B=\tr_A\op{\psi_k}{\psi_k}$ have the same eigenvalue spectrum, they are unitarily related.  This means we have $\rho_k^B=\frac{1}{2}(\mbb{I}+\mbf{r}'_k\cdot\vec{\sigma})$ for Bloch vectors $\mbf{r}'_k$ for which $|\mbf{r}'_k|=|\mbf{r}_k|$.  Similarly, we can write
\begin{align}
    \sigma_k^B=\frac{1}{q_k}B_k B_k^{\dagger}= \frac{1}{2}(\mbb{I}+\mbf{s}_k\cdot\vec{\sigma}),
\end{align}
where $q_k=\tr[B_kB_k^\dagger]$ and $\sigma_k^A=\frac{1}{2}(\mbb{I}+\mbf{s}_k'\cdot\vec{\sigma})$ with $|\mbf{s}'_k|=|\mbf{s}_k|$.  Expressing the commutators in terms of the Bloch vectors,  
\begin{align}
\label{Eq:Bloch-vector-commutator-A}
    \lbrack \rho_1^A, \rho_2^A \rbrack
    &=\frac{i}{2}(\mbf{r}_2\times\mbf{r}_1)\cdot\vec{\sigma}\\
    \label{Eq:Bloch-vector-commutator-B}
    \lbrack \sigma_1^B, \sigma_2^B \rbrack
    &=\frac{i}{2}(\mbf{s}_2\times\mbf{s}_1)\cdot\vec{\sigma},
\end{align}
the vanishing commutators imply that $\{\mbf{r}_1,\mbf{r}_2\}$ and $\{\mbf{s}_1,\mbf{s}_2\}$ are pairs of vectors that are either parallel (denoted $\uparrow\uparrow)$ or anti-parallel (denoted $\uparrow\downarrow$).  We identify the pair $\{\mbf{r}_1,\mbf{r}_2\}$ ($\{\mbf{s}_1,\mbf{s}_2\}$) as being parallel when one of them is zero.

The reduced density operators of the given state $\ket{\psi}^{ABC}$ can also be expressed in terms of their Bloch vectors,
\begin{align}
    \rho^{A} &= \tr_{BC} \op{\psi}{\psi}^{ABC}= \frac{1}{2}(\mbb{I}+\mbf{R}\cdot\vec{\sigma})\\
    \quad\sigma^{B} &= \tr_{AC} \op{\psi}{\psi}^{ABC}= \frac{1}{2}(\mbb{I}+\mbf{S}\cdot\vec{\sigma}).
\end{align}
It is straightforward to compute that
\begin{align}
   E^{A|BC}(\psi)&= 1 - |\mbf{R}|,& E^{B|AC}(\psi)&= 1 - |\mbf{S}|.
\end{align}
Now suppose that $\mbf{r}_1\uparrow\uparrow\mbf{r}_2$.  Then since $\mbf{R}=\sum_k p_k\mbf{r}_k$, when Charlie measures $\ket{\psi}$ in the basis $\ket{k}^C$, the average post-measurement $E_2$ entanglement for Alice and Bob will be
\begin{align}
    \langle E_2^{A|B} \rangle = 1 - \sum_k p_k |\mbf{r}_k| = 1 - |\mbf{R}| = E_2^{A|BC}.
\end{align}
Similarly, when $\mbf{s}_1\uparrow\uparrow\mbf{s}_2$, Charlie can measure in the basis $\ket{k'}^C$ to assist in generating entanglement $E_2^{B|AC}$ for Alice and Bob.  Therefore, we have proven Theorem \ref{Thm:1} whenever $\mbf{r}_1\uparrow\uparrow\mbf{r}_2$ or $\mbf{s}_1\uparrow\uparrow\mbf{s}_2$.  It remains to consider the case when both $\mbf{r}_1\uparrow\downarrow\mbf{r}_2$ and $\mbf{s}_1\uparrow\downarrow\mbf{s}_2$ with none of these Bloch vectors being zero.

Suppose that $\mbf{r}_1\uparrow\downarrow\mbf{r}_2$.  Since 
\[\sigma^B=\tr_{AC}\op{\psi}{\psi}=\sum_k p_k\rho_k^B=\sum_k \frac{p_k}{2}(\mbb{I}+\mbf{r}_k'\cdot\vec{\sigma}),\]
we have
\begin{align}
\label{Eq:Block-triangle-ineq}
    |\mbf{S}|=|p_1\mbf{r}_1'+p_2\mbf{r}_2'|&\geq |p_1|\mbf{r}_1'|-p_2|\mbf{r}'_2||\notag\\
    &=|p_1|\mbf{r}_1|-p_2|\mbf{r}_2||\notag\\
    &=|p_1\mbf{r}_1+p_2\mbf{r}_2|=|\mbf{R}|,
\end{align}
where the first inequality is the triangle inequality and the last line follows from the vectors being anti-parallel.  Similarly, since
\[\rho^A=\tr_{BC}\op{\psi}{\psi}=\sum_k q_k\sigma_k^A=\sum_k \frac{q_k}{2}(\mbb{I}+\mbf{s}_k'\cdot\vec{\sigma}),\] the assumption that $\mbf{s}_1\uparrow\downarrow\mbf{s}_2$ 
yields $|\mbf{R}|\geq|\mbf{S}|$. Hence, we have that $|\mbf{R}|=|\mbf{S}|$ or equivalently $E_2^{A|BC}(\psi)=E_2^{B|AC}(\psi)$.  Furthermore, the equality condition in Ineq. \eqref{Eq:Block-triangle-ineq} shows that $\mbf{r}'_1\uparrow\downarrow\mbf{r}'_2$, and so we obtain the commutator relation $[\rho_1^B,\rho_2^B]=0$.  Therefore, when comparing with Eq. \eqref{Eq:Bloch-vector-commutator-A}, we see that the reduced density matrices for $\ket{\psi_1}$ and $\ket{\psi_2}$ are pairwise commuting for both parties.  Since $\mbf{r}'_1$ and $\mbf{r}_2'$ are nonzero, the reduced density matrices have no degeneracy in their eigenvalue spectrum, and so the states $\ket{\psi_1}$ and $\ket{\psi_2}$ have unique Schmidt bases for both parties.  The pairwise commuting relations then imply that either $\ket{\psi}\overset{\mathrm{LU}}{\simeq}(a\ket{00}+b\ket{11})\ket{0}+(c\ket{01}+d\ket{10})\ket{1}$ or $\ket{\psi}\overset{\mathrm{LU}}{\simeq}(a\ket{00}+b\ket{11})\ket{0}+(c\ket{00}+d\ket{11})\ket{1}$.  In the first case, the condition $E_2^{A|BC}=E_2^{B|AC}$ requires that either $a=b$ or $c=d$, which is a contradiction since we assume that the Bloch vectors are nonzero.  In the second case, we can rewrite the state as
\begin{equation}\label{Eq:Lemma_1_Form}
    \ket{\psi}^{ABC}\overset{\mathrm{LU}}{\simeq}\sqrt{p}\ket{00}\ket{\eta_0}+\sqrt{1-p}\ket{11}\ket{\eta_1},
\end{equation}
where the $\ket{\eta_0}$ and $\ket{\eta_1}$ on Charlie's system are normalized but not necessarily orthogonal. However, we can always find a basis $\{\ket{e_0},\ket{e_1}\}$ such that
\begin{align}
    \ket{\eta_0}&=\cos\theta\ket{e_0}+\sin\theta\ket{e_1}\notag\\
    \ket{\eta_1}&=\cos\theta\ket{e_0}-\sin\theta\ket{e_1}
\end{align}
for some $\theta$ determined by $\ip{\eta_0}{\eta_1}$. Then
\begin{equation}
    \begin{aligned}
    \ket{\psi}^{ABC}\overset{\mathrm{LU}}{\simeq}&\cos\theta(\sqrt{p}\ket{00}+\sqrt{1-p}\ket{11})\ket{e_0}\\
    +&\sin\theta(\sqrt{p}\ket{00}-\sqrt{1-p}\ket{11})\ket{e_1}.
    \end{aligned}
\end{equation}
When Charlie measures in the basis $\ket{e_k}$, we get the optimal post-measurement entanglement.
\end{proof}

We now ask for conditions under which we achieve lossless decoupling when EoA is measured with respect to other entanglement monotones. As noted above, any two-qubit pure state entanglement monotone $E$ is given by $E(\varphi^{AB})=f(\lambda_{\min}(\varphi^{A}))$ with  $f:[0,\frac{1}{2}]\rightarrow \mbb{R}^+$ being concave and $\lambda_{\min}(\varphi^A)$ the smallest eigenvalue of $\varphi^A$.  Via Eq. \eqref{Eq:EOA-defns}, there exists an optimal ensemble $\{p_x,\ket{\varphi_x}\}_x$ (with all $p_x>0$) for $\tr_C{\op{\psi}{\psi}}^{ABC}$ such that
\begin{align}
   E^{(a)}(\psi^{ABC})&=\sum_x p_xE(\varphi_x^{AB})\notag\\
  &=\sum_x p_xf(\lambda_{\min}(\varphi_x^{A}))\notag\\
 &\overset{\text{(a)}}{\leq} f\left(\sum_x  p_x\lambda_{\min}(\varphi_x^A)\right)\notag\\
 &\overset{\text{(b)}}{\leq} f(\lambda_{\min}(\psi^A))=E^{A|BC}(\psi^{ABC}).\label{Eq:chain-ineq}
\end{align}
Inequality (a) is the concavity of $f$, and (b) follows from $f$ monotonically increasing on the interval $[0,\frac{1}{2}]$ along with the fact that 
\begin{equation}
\label{Eq:lambda-min-inequ}
\sum_x p_x\lambda_{\min}(\varphi_x^A)\leq \lambda_{\min}\left(\sum_x p_x\varphi_x^A\right)=\lambda_{\min}(\psi^A).
\end{equation}
A similar chain of inequalities allows us to replace the RHS with $E^{B|AC}(\psi^{ABC})$.  

Motivated by the fact that some common entanglement monotones such as entanglement entropy and concurrence are determined by an $f$ \textit{strictly} concave, we restrict $f$ to be as such. In this case, $f$ is then continuous and strictly increasing on the interval $(0,1/2]$.  If we find for example that $E^{(a)}(\psi^{ABC})=E^{A|BC}(\psi^{ABC})$, then the strict concavity of $f$ implies that the optimal ensemble $\{p_x,\ket{\varphi_x}\}_x$ in Eq. \eqref{Eq:chain-ineq} either consists of just a single nonzero element, or $\lambda_{\min}(\varphi_x^A)$ is the same for every $\varphi_x$.  If there is just a single element in the pure-state decomposition of $\varphi^{AB}$, then system $C$ is already decoupled. Otherwise, if the ensemble has multiple elements with $\lambda_{\min}(\varphi_x^A)$ being the same for every $\varphi_x$,  then because these are qubit operators, there exist unitary matrices $U_x$ such that
\begin{equation}
\varphi_x^A=U_x\varphi_0^A U_x^\dagger
\end{equation}
for some fixed $\varphi_0^A$.  Tightness in Ineq. \eqref{Eq:lambda-min-inequ} then requires 
\begin{equation}
\label{Eq:lambda-min-eq}
\lambda_{\min}(\varphi_0^A)=\lambda_{\min}\left(\sum_x p_xU_x \varphi_0^A U_x^\dagger \right)=\lambda_{\min}(\psi^A),
\end{equation}
and the following proposition characterizes when such a condition holds. 
\begin{proposition}
\label{Prop:lambda-min-commuting}
A qubit Hermitian operator $H$ satisfies $\lambda_{\min}(H)=\lambda_{\min}\left(\sum_{x=0}^t p_x U_x H U_x^\dagger\right)$ for unitary operators $U_x$ such that $U_0=\mbb{I}$ and non-zero probabilities $p_x$ if and only if $[H,U_x]=0$ for all $x$.
\end{proposition}
\begin{proof}
(Here we present a simple proof based on an external result. A slightly longer but from first principles one is presented in the appendix). Sufficiency is obvious. To prove necessity, note that the operators $\sqrt{p_x}U_x$ described above may be thought of as Kraus operators for a unital qubit channel. In \cite{Kribs}, it is shown that the fixed point set of a unital channel is exactly the set of all operators commuting with all Kraus operators in any operator-sum representation of the channel. With this in mind, suppose that $\lambda_{\min}(H)=\lambda_{\min}\left(\sum_x^t p_x U_x H U_x^\dagger\right)$. Then we have for some unitary $V$, $V^{\dagger}HV=\sum_x^t p_x U_x H U_x^\dagger$, so that $H$ is a fixed point of a new unital channel whose Kraus operators are $\sqrt{p_x}VU_x$. From here it follows that $[H, VU_x]=0$ for all $x$. The requirement that $U_0=\mbb{I}$ implies $[H,V]=0$, so that $[H, U_x]=0$ for all $x$.
\end{proof}

We combine this proposition with Eq. \eqref{Eq:lambda-min-eq} to conclude that
\begin{theorem}
\label{thm:2}
Let $E$ be an entanglement monotone for two qubits determined by a strictly concave function of the Schmidt coefficients.  Then a three qubit state $\ket{\psi}^{ABC}$ will satisfy $E^{(a)}(\psi)=E^{A|BC}(\psi)$ if and only if
\begin{align}
\label{Eq:three-qubit-equality-decomposition}
    \ket{\psi}^{ABC}\overset{\mathrm{LU}}{\simeq} \sum_{x}\sqrt{p_x} (U_x^A\otimes V_x^B)\ket{\lambda_{\min}}^{AB}\ket{x}^{C}
\end{align}
where the $p_x$ are understood to be non-zero, $\ket{\lambda_{\min}}=\sqrt{\lambda_{\min}}\ket{00}+\sqrt{1-\lambda_{\min}}\ket{11}$ for $\lambda_{\min}:=\lambda_{\min}(\psi^A)$, and either each $U_x^A$ is diagonal in the $\{\ket{0},\ket{1}\}$ basis or $\lambda_{\min}=1/2$.  An analogous statement holds under an interchange of systems $A$ and $B$.
\end{theorem}

The statement of Theorem \ref{thm:2} also holds for $2\otimes2\otimes n$ states, as there is nothing restricting the dimension of Charlie's system in the preceding analysis. Note that in any case, it suffices to consider an index $x$ in Eq. \eqref{Eq:three-qubit-equality-decomposition} that has a range of no more than four distinct values \cite{Uhlmann-1998a}.  

When Charlie's system is a qubit, it is instructive to compare Eq. \eqref{Eq:three-qubit-equality-decomposition} with other three-qubit decompositions known in the literature.  Wootters has shown that every two-qubit mixed state $\rho^{AB}$ admits a pure-state decomposition with each state having the same concurrence \cite{Wootters-1998a}.  Hence every three-qubit state $\ket{\psi}$ can be written as
\begin{align}
\label{Eq:three-qubit-equality-decomposition-wootters}
    \ket{\psi}^{ABC}\overset{\mathrm{LU}}{\simeq} \sum_{x}\sqrt{p_x} (U_x^A\otimes V_x^B)\ket{\lambda}^{AB}\ket{x}^{C}
\end{align}
where $\ket{\lambda}=\sqrt{\lambda}\ket{00}+\sqrt{1-\lambda}\ket{11}$ and the local unitaries are unrestricted. De Vicente \textit{et al.} later showed \cite{deVicente-2012a} that for three-qubit pure states, a decomposition always exists of the form
\begin{align}
\label{Eq:three-qubit-equality-decomposition-deVicente}
    \ket{\psi}^{ABC}\overset{\mathrm{LU}}{\simeq} \frac{1}{\sqrt{2}}\left(\ket{\lambda}^{AB}\ket{1}^C+U^A\otimes V^B\ket{\lambda}^{AB}\ket{2}^C\right),
\end{align}
where again the local unitaries are unrestricted. The crucial difference between our Eq. \eqref{Eq:three-qubit-equality-decomposition} and these two decompositions is that $\lambda=\lambda_{\min}(\psi^A)$ and the unitaries $U_x^A$ are required to commute with $\tr_B\op{\lambda}{\lambda}^{AB}$.  Only a special subset of states $\ket{\psi}$ will admit such a decomposition, and so most states satisfy $E^{(a)}(\psi)< \min\{E^{A|BC}(\psi),E^{B|AC}(\psi)\}$.  In contrast, Theorem \ref{Thm:1} shows that $E_2^{(a)}(\psi)= \min\{E_2^{A|BC}(\psi),E_2^{B|AC}(\psi)\}$ for all three-qubit states. The reason that we have universality for $E_2$ is that it is an entanglement monotone not based on a strictly concave eigenvalue function, and in fact, inequality (a) in \eqref{Eq:chain-ineq} will always be tight since $f$ is linear in $\lambda_{\min}$ for $E_2$.

As a final remark, one might wonder whether or not Theorem \ref{thm:2} also applies to the entanglement of collaboration.  That is, whether $E^{(c)}(\psi)=E^{A|BC}(\psi)$ if and only if $\ket{\psi}$ has the form of Eq. \eqref{Eq:three-qubit-equality-decomposition}.  In fact it does.  To see this,  suppose that $E^{(c)}(\psi)=E^{A|BC}(\psi)$ and consider a general multi-outcome LOCC transformation $\ket{\psi}^{ABC}\to\ket{\varphi_i}$ that attains an average post-measurement entanglement of $E^{A|BC}(\psi)$ (a similar argument holds for average post-measurement entanglement of $E^{B|AC}(\psi)$).  Following the reasoning of Theorem \ref{thm:2}, suppose that Bob performs a measurement that induces the transformation $\ket{\psi}^{ABC}\to \ket{\psi_x}^{ABC}$, with $\ket{\psi_x}^{ABC}$ having probability $p_x$.  Note that the $A|BC$ entanglement cannot decrease on average.  Then by strict concavity of $f$, we must have that $\lambda_{\min}(\psi_x^A)$ is the same for every $\psi_x^{ABC}$.  Hence Eq. \eqref{Eq:lambda-min-eq} holds and the conclusion of Proposition \ref{Prop:lambda-min-commuting} implies that $\psi^A=\psi_x^A$ for all $x$.  In other words, if Bob's measurement is described by Kraus operators $\{M_x\}_x$, then we have that
    \[\tr_B[\mbb{I}^A\otimes(\mbb{I}-p_x^{-1}M_x^\dagger M_x)^B\psi^{AB}]=0 \]
for all outcomes $x$ with nonzero probability.  This is equivalent to the condition
\begin{equation}
\label{Eq:vanishing-trace}
   \sqrt{\psi^A}(\mbb{I}-p_x^{-1}M_x^* M_x^T)\sqrt{\psi^A}=0.
\end{equation}
Since the rank of $\psi^A$ is two (or else $\ket{\psi}^{ABC}$ would be a product state), $\psi^A$ is invertible and we thus obtain $\mbb{I}=p_x^{-1}M_x^\dagger M_x$, which means that Bob's measurement is the application of a random unitary.  The same reasoning applies for a measurement by Alice when interpreting $\ket{\psi}^{A(BC)}$ as a two-qubit state; i.e. $\psi^{BC}$ can be expressed as a $2\times 2$ positive matrix in the Schmidt basis of $\ket{\psi}^{A(BC)}$.  In summary, a general LOCC protocol that converts $\ket{\psi}^{ABC}\to\ket{\varphi_i}$ can be replaced by one in which Alice and Bob only apply random local unitaries.  Since these local unitaries can always be delayed until the end of the protocol, we conclude that
\begin{align}
    E^{(c)}(\psi)&=\min\{E^{A|BC}(\psi),E^{B|AC}(\psi)\} \notag\\
    \Rightarrow E^{(a)}(\psi)&=E^{(c)}(\psi).
\end{align}
The argument given here is a variation of the one given in Ref. \cite{Bennett-2000a} for entropy-preserving LOCC transformations.

Finally, we close by a general observation that follows from the analysis carried out in Theorem \ref{Thm:1}.  We suspect this might have independent interest in the study of three-qubit networks. If one labels the vertices of a graph by the bipartite pure state concurrences of $\ket{\psi}^{ABC}$, and the edges by the mixed state concurrences as shown in Fig. \ref{fig:symmetry}, a symmetry relationship follows which is made precise in the following corollary.

\begin{figure}[h]
\includegraphics{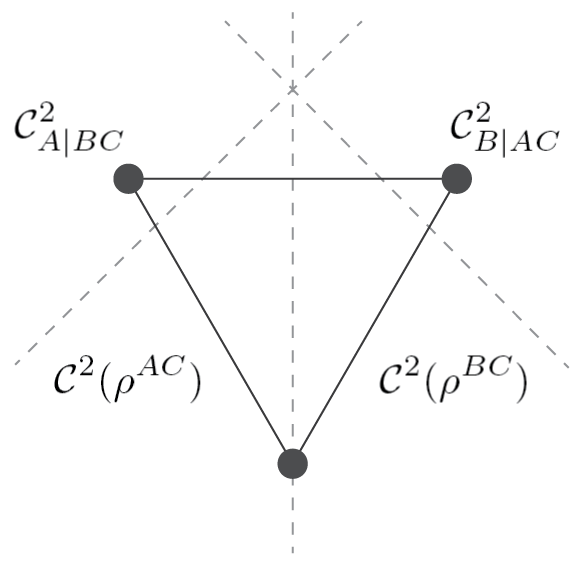}
\caption{Symmetry of entanglement across bipartite ``cuts" implies symmetry of pairwise entanglement. For almost every state, it also implies symmetry under exchange of the respective parties.}
\label{fig:symmetry}
\end{figure}

\begin{corollary}
\label{Cor:final}
For any three-qubit pure state $\ket{\psi}^{ABC}$ such that $\rho_k^A\neq\frac{1}{2}\mbb{I}$ for $k=1,2$, the following are equivalent.
\begin{equation}
    \begin{aligned}
    \text{(i)}&\quad E_2^{A|BC}=E_2^{B|AC}\\
    \text{(ii)}&\quad \ket{\psi}^{ABC}\overset{\mathrm{LU}}{\simeq} \ket{\psi}^{BAC}\\
    \text{(iii)}&\quad \mc{C}(\rho^{AC})=\mc{C}(\rho^{BC})
    \end{aligned}
\end{equation}
where $\mc{C}(\rho)$ is the two-qubit mixed state concurrence \cite{PhysRevA.54.3824}. 
\end{corollary}

\begin{proof}
If Charlie is decoupled from Alice and Bob, then Alice and Bob can always swap their state via LU's, which follows from Schmidt decomposition. If Charlie is not decoupled, we can see from the analysis in Theorem \ref{Thm:1} that $E_2^{A|BC}=E_2^{B|AC}$ if and only if both $\lbrack \rho_1^A, \rho_2^A \rbrack$ and $\lbrack \rho_1^B, \rho_2^B \rbrack$ vanish. The additional restriction that none of these reduced density matrices is $\frac{1}{2}\mbb{I}$ implies that $\ket{\psi}^{ABC}$ has the form of Eq. \eqref{Eq:Lemma_1_Form}, which is symmetric under exchange of Alice and Bob. Therefore (i) implies (ii), and the converse is trivial. Equivalence of (i) and (iii) follows from the CKW monogamy relations \cite{Coffman-2000a}
\begin{align}\label{ckw}
    \mc{C}^2_{A|BC} &= \tau_{ABC} + \mc{C}^2(\rho^{AB})+\mc{C}^2(\rho^{AC})\\
    \mc{C}^2_{B|AC} &= \tau_{BAC} + \mc{C}^2(\rho^{BA})+\mc{C}^2(\rho^{BC}).
\end{align}
In particular, the three-tangle $\tau$ and the mixed state concurrence are symmetric under exchange of parties, so taking the difference of the above equations gives
\begin{align}
    \mc{C}^2(\rho^{AC})-\mc{C}^2(\rho^{BC})= \mc{C}^2_{A|BC}-\mc{C}^2_{B|AC}.
\end{align}
The pure state concurrence is however in one-to-one correspondence with $\lambda_{\text{min}}$, so equivalence of (i) and (iii) follows.
\end{proof}

We note here that the restriction $\rho_k^A\neq\frac{1}{2}\mbb{I}$ is not actually very restrictive because when choosing three-qubit states at random, $\rho_k^A=\frac{1}{2}\mbb{I}$ occurs with probability zero. Nonetheless, there are certain states with $\rho_k^A=\frac{\mbb{I}}{2}$ for some $k$ that can be included in the statement of Corollary \ref{Cor:final}. For example, any state of the form 
\begin{equation}
    \ket{\psi}^{ABC}\overset{\mathrm{LU}}{\simeq}\sqrt{p}\ket{\Phi^+}\ket{0}+\sqrt{1-p}\ket{\varphi}\ket{1},
\end{equation}
with $\ket{\varphi}^{BA}=(U\otimes U^*)\ket{\varphi}^{AB}$ for some unitary $U$ will satisfy the conclusion of Corollary \ref{Cor:final}.
\section{Conclusion and Discussion}

In Theorem \ref{Thm:1} we have established for three-qubit pure states that imposing locality constraints on a helper system does not decrease the probability of obtaining a maximally entangled state $\ket{\Phi^+}$ between any two fixed parties. One may wonder whether this holds for tripartite systems in other dimensions.  It can already be seen that it does not when Alice and Bob have qutrit ($d=3$) systems.  Consider density matrices $\rho^{AB}$ in which $\rho^A$ and $\rho^B$ are both maximally mixed; i.e. $\frac{\mbb{I}}{3}$.  Then $\rho^{AB}$ is the Choi matrix for a qutrit unital channel; however, it is known that not every qutrit unital channel can be expressed as a probabilistic mixture of unitary operations \cite{Landau-1993a}.  This means that it is not possible to find a pure-state decomposition of $\rho^{AB}$ consisting of only maximally entangled states, which would be required if Ineq. \eqref{Eq:EoA-bipartite-bound} were tight.  For general $2\otimes 2\otimes n$ states with $n>2$, we leave it as an open problem to decide whether $E^{(a)}_2=\min\{E_2^{A|BC}, E_2^{B|AC}\}$, but we conjecture it to be true.

% To see why not, suppose that $E^{(a)}_2=E_2^{A|BC}$. This means there exists an ensemble for $\varphi^{AB}$ such that $\lambda_{\text{min}}(\varphi^A)=\sum_x p_x \lambda_{\text{min}}(\varphi_x^A)$ where the $p_x$ are understood to be non-zero. In accordance with statements in the proof of Prop. \ref{Prop:lambda-min-commuting}, we must have $\varphi^A$ and $\varphi_x^A$ for all $x$ sharing common eigenvectors. Therefore, when none of $\lambda_{\text{min}}(\varphi_x^A)$ are $1/2$, it is necessary that 
% \begin{align}\label{Eq:General-n}
%     \ket{\psi}^{ABC}\overset{\mathrm{LU}}{\simeq}\sum_{x} \sqrt{p_x}(\mbb{I}^A\otimes V^B_x)\ket{\lambda_x}^{AB}\ket{x}^C
% \end{align}
% where $\ket{\lambda_x}=\sqrt{\lambda_x}\ket{00}+\sqrt{1-\lambda_x}\ket{11}$ for $\lambda_x :=\lambda_{\text{min}}(\varphi_x^A)$ and the $V^B_x$ are arbitrary unitaries. When it is found that $\lambda_x=1/2$, a unitary $U^A_x$ can be appended to this term in the sum without disturbing optimality.

Note that Theorem \ref{Thm:1} can be phrased entirely in terms of the two-qubit density matrix as follows: every rank-two, two-qubit density matrix $\rho^{AB}$ satisfies
\begin{equation}
\label{Eq:EoA-density}
    \wt{E}^{(a)}_2(\rho^{AB})=2\min\{\lambda_{\min}(\rho^A),\lambda_{\min}(\rho^B)\}.
\end{equation}
Interestingly this relationship also applies for any function that varies linearly in $\lambda_{\min}$ over the interval $(0,\frac{1}{2}]$. As an extreme case, consider the monotone 
\begin{equation}
   S_0(\lambda) = \begin{cases}
   0 & \text{if}\quad \lambda=0 \\
   1 & \text{if}\quad 0<\lambda\leq \frac{1}{2}, \end{cases}
\end{equation}
which is the $\alpha$-entropy of entanglement for $\alpha=0$. This function also satisfies Eq. \eqref{Eq:EoA-density}.  To see this, observe that it trivially holds for any $\rho^{AB}$ in which either reduced state $\rho^{A}$ or $\rho^B$ is pure.  On the other hand, if both $\rho^{A}$ and $\rho^B$ are genuinely mixed, then it is always possible to find a pure-state decomposition for $\rho^{AB}$ consisting of entangled states.  This follows from the fact that in such a case the set of product states in the support of $\rho^{AB}$ is measure zero.  We give a rigorous argument of this in the appendix.

\section*{Acknowledgment}

This work was supported by NSF Award No. 1914440.

\section{Appendix}

\subsection{Alternative Proof of Prop. \ref{Prop:lambda-min-commuting}}

Sufficiency is obvious. To prove necessity, we note that for Hermitian operators $A$ and $B$, $\lambda_{\text{max}}(A+B)=\lambda_{\text{max}}(A) + \lambda_{\text{max}}(B)$ if and only if $A$, $B$, and $A+B$ share a principle eigenvector which has the maximal eigenvalue for $A$, $B$, and $A+B$ \cite{Bhatia-1997a}. We can extend this by induction to a finite set of operators, taking them to be the $p_x U_x H U_x^{\dagger}$. Because these are qubit operators and $\tr{(\sum_x^t p_x U_x H U_x^\dagger)}=\tr{(H)}$, we can replace max with min in the above fact. Therefore suppose that $\lambda_{\min}\left(\sum_x^t p_x U_x H U_x^\dagger\right)=\lambda_{\min}(H)$.  Then there is a non-zero vector $\ket{h}$ such that $\forall\ x$,
\begin{equation}
\begin{aligned}
    H\ket{h}&=\lambda_{\max}(H)\ket{h}\\
    U_xHU_x^{\dagger}\ket{h}&=\lambda_{\max}(H)\ket{h},
\end{aligned}
\end{equation}
where the first line follows from the assumption that $U_0=\mbb{I}$.  These jointly imply that
$[H,U_x^{\dagger}] \ket{h}=0.$ Let $\ket{h^\perp}$ be a vector orthogonal to $\ket{h}$, which is uniquely determined up to an overall phase.  Then, $\ket{h^\perp}$ is also an eigenvector of $H$ and the $U_xHU_x^{\dagger}$ (with eigenvalue possibly zero), and for the same reason as above, $[H,U_x^{\dagger}] \ket{h^\perp}=0.$  Hence
\begin{equation}
     [H,U_x^{\dagger}] (\alpha\ket{h}+\beta\ket{h^\perp})=0
\end{equation}
for arbitrary $\alpha,\beta$, and so we have $[H,U_x]=0$ for all $x$.

\subsection{Decomposition of Certain Density Matrices into Entangled States}

Consider a two-qubit separable state $\rho^{AB}=\sum_{k=1}^4 \op{\alpha_k}{\alpha_k}^A\otimes\op{\beta_k}{\beta_k}^B$.  Suppose that both $\rho^A$ are $\rho^B$ are genuinely mixed.  Then we can always find a pair $\ket{\alpha_{k_1}}\ket{\beta_{k_1}}$ and  $\ket{\alpha_{k_2}}\ket{\beta_{k_2}}$ such that $\{\ket{\alpha_{k_1}},\ket{\alpha_{k_2}}\}$ and $\{\ket{\beta_{k_1}},\ket{\beta_{k_2}}\}$ span two-dimensional spaces.  This means that the two-dimensional subspace of $\mbb{C}^2\otimes\mbb{C}^2$ subspace spanned by $\{\ket{\alpha_{k_1}}\ket{\beta_{k_1}},\ket{\alpha_{k_2}}\ket{\beta_{k_2}}\}$ contains at most two product states \cite{Sanpera-1998a}.  Hence, we can always find a $2\times 2$ unitary matrix with components $u_{ij}$ such that $u_{11}\ket{\alpha_{k_1}}\ket{\beta_{k_1}}+u_{12}\ket{\alpha_{k_2}}\ket{\beta_{k_2}}$ and $u_{21}\ket{\alpha_{k_1}}\ket{\beta_{k_1}}+u_{12}\ket{\alpha_{k_2}}\ket{\beta_{k_2}}$ are both entangled states.  We can thus replace product states $k_1$ and $k_2$ in the decomposition of $\rho^{AB}$ with these entangled states.  We then choose a remaining product state $\ket{\alpha_{k_3}}\ket{\beta_{k_3}}$ in the ensemble and mix it with one of the entangled states by another $2\times 2$ unitary matrix so to generate two new entangled states.  This allows us to replace $\ket{\alpha_{k_3}}\ket{\beta_{k_3}}$ with an entangled state in the decomposition of $\rho^{AB}$.  Proceeding in this way allows us to eliminate all product states in the decomposition.  Note, this argument also applies for separable states in arbitrary dimensions since a generic two dimensional subspace in $\mbb{C}^{d_A}\otimes\mbb{C}^{d_B}$ will always contain a finite number of product states.

\bibliography{EoAbib.bib}

\end{document}